\documentclass[12pt,english]{article}
\usepackage[english]{babel}
\usepackage{amsmath,amssymb,amscd,color,amsthm}
\usepackage{amsfonts}
\usepackage{graphicx}
\usepackage{url}
\usepackage{hyperref}
\usepackage{cite}

\makeatletter
\renewcommand\section{\@startsection {section}{1}{\z@}%
                                   {-5.5ex \@plus -1ex \@minus -.2ex}
                                   {2.3ex \@plus.2ex}%
                                   {\normalfont\large\bfseries}}
\renewcommand\subsection{\@startsection{subsection}{2}{\z@}%
                                     {-3.25ex\@plus -1ex \@minus -.2ex}%
                                     {1.5ex \@plus .2ex}%
                                     {\normalfont\bfseries}}

\numberwithin{equation}{section}

\makeatother

\newcommand{\bea}{\begin{eqnarray}}
\newcommand{\eea}{\end{eqnarray}}
\newcommand{\be}{\begin{equation}}
\newcommand{\ee}{\end{equation}}

\newcommand{\Z}{{\mathbb Z}}
\newcommand{\R}{{\mathbb R}}

\def\eg{{\it e.g.~}}

\newcommand{\cF}{{\cal F }}

\newcommand{\ie}{{\it i.e.~}}

\renewcommand{\title}[1]{\vbox{\center\LARGE{#1}}\vspace{5mm}}
\renewcommand{\author}[1]{\vbox{\center#1}\vspace{5mm}}
\newcommand{\address}[1]{\vbox{\center\footnotesize\em#1}}
\newcommand{\email}[1]{\vbox{\center\footnotesize\tt#1}\vspace{5mm}}

\theoremstyle{definition}\newtheorem{prop}{Proposition}
\theoremstyle{definition}
\newtheorem{defn}{Definition}[section]
\newtheorem{example}{Example}[section]
\newtheorem{lem}{Lemma}[section]

\begin{document}

\begin{titlepage}

 \begin{flushright}

\end{flushright}

\begin{center}

\hfill \\
\hfill \\
\vskip 1cm

\title{Sparse Modular Forms, Lattices, and Codes}

\author{Christoph A. Keller$^{a,b}$, Ashley Winter Roberts$^a$}

\address{
${}^a$ Department of Mathematics, University of Arizona, Tucson, AZ 85721-0089, USA
\\
${}^b$ Institut f\"ur Theoretische Physik, ETH Zurich,
CH-8093 Z\"urich, Switzerland}

\email{ cakeller@arizona.edu, awroberts@arizona.edu}

\end{center}

\vfill

\abstract{ Motivated by sparseness conditions for holographic CFTs, we investigate sparseness of modular forms, lattices, and codes. For this we investigate the free energy of such objects as their weight, dimension or size goes to infinity.
We construct families of modular forms that are sparse, such as the Eisenstein series $E_{2k}(\tau)$. We then investigate lattices that come from codes and introduce a sparseness condition for such lattices. We investigate the limit of lattices constructed from self-dual Reed-Muller codes and prove that they are sparse in this sense.}

\vfill

\end{titlepage}

\eject

\tableofcontents

\section{Introduction}

Holographic CFTs that describe gravity on anti-de Sitter spaces need to satisfy various conditions. One such condition is usually some type of sparseness condition: the light spectrum should not contain too many states. On the most basic level, this is necessary for a large $N$ or large central charge limit to exist. If the number of light states diverges as $N$ or $c$ goes to infinity, such a limit does not exist.

On top of that, one usually wants to match a specific spectrum on the gravity side. This requires imposing additional conditions. On AdS${}_3$ for instance, one could for instance assume the existence of pure gravity. This theory only has boundary gravitons and BTZ black holes as degrees of freedom. Its holographic dual would then be an extremal CFT, that is a CFT that only has Virasoro descendants up to the weight where additional primaries are required from modular invariance \cite{Witten:2007kt}.

This is the strongest sparseness condition that one could try to impose. In fact it is unclear both on the gravity and the CFT side if such a theory exists. A weaker condition one can impose is to require that the free energy has the right thermodynamic properties; namely that it exhibits a Hawking-Page phase transition \cite{Hawking:1982dh}. Such a transition occurs in Einstein gravity, where the dominating contribution transitions from the vacuum at low temperature to BTZ black holes at high temperature. However, if a theory has a lot of light states, their contribution might dominate at intermediate temperature and hence smooth out the transition.

To avoid this, it is necessary to impose a sparseness condition on the light states. This was worked out in \cite{Hartman:2014oaa, Dey:2024nje}. Even though this condition is much less restrictive than the extremality condition above, in practice it is still not that easy to satisfy. There are actually not as many explicit examples of CFTs that satisfy it as one might expect. Most examples are symmetric orbifolds or closely related constructions.

In particular, families of lattice CFTs, that is Narain constructions, do not satisfy this sparseness condition. In fact, they do not even have a large $c$ limit due to the fact that the number of free boson descendants diverges. Nonetheless it is of interest to consider such families. It therefore makes sense to introduce new types of sparseness conditions that are appropriate for such theories.
Such a condition was indeed introduced in \cite{Dymarsky:2020pzc}. It guarantees a similar behavior for the growth of entropy as the original sparseness condition.
However, no explicit families of Narain CFTs satisfying this condition are known yet. One of the goals of this article is to make progress in this direction.

There is a very natural candidate for such sparse lattices: extremal lattices. These are even self-dual lattices, \ie lattices that lead to holomorphic CFTs, whose shortest vector has length square $2(\lfloor \frac{d}{24}\rfloor+1)$. This is in fact the longest shortest vector that is compatible with modularity, similar to the condition of extremal CFTs. To put it another way, they are lattices whose theta function is given by the unique extremal modular form of weight $d/2$. Such lattices have been constructed up to $d=80$ \cite{MR209983,MR1662447,MR1489922,MR2999133,MR3225314}. However, they are known not to exist if $d>163264$ \cite{MR0376536,MR2854563}. The reason for this is that extremal modular forms start to have negative coefficients at that weight and can therefore no longer correspond to lattice theta functions. 
Nonetheless, as a warmup we establish that extremal modular forms are indeed sparse in the sense of \cite{Dymarsky:2020pzc}. 

We then turn to Eisenstein series $E_k(\tau)$. Unlike extremal modular forms, for doubly even $k$, these always have non-negative coefficients. We establish that they are sparse. However, for $k>6$ they no longer have integer coefficients. Therefore, they  cannot directly correspond to the theta function of a lattice. However, by the Siegel-Weil formula \cite{MR165033,MR223373} they correspond to a suitable average of the theta function over all even self-dual lattices of a given dimension \cite{MR2209183}. The conclusion is thus that the `average' (in an appropriate sense) chiral lattice CFT is sparse, as was already pointed out in \cite{Dymarsky:2020pzc}.

This leads to the somewhat vexing situation that even though the average lattice CFT should be sparse, we know no concrete example. To resolve this, we turn to the question of constructing examples.
It has been known for a long time that lattices can be constructed from linear error-correcting codes. There are various constructions of this type, for instance Construction A \cite{MR285994}.
There are also constructions B,C,D, which give lattices with different properties starting from different types of codes \cite{MR1662447}.
These constructions all give even Euclidean lattices that can then be turned into CFTs by the standard constructions. More precisely, the lattices themselves give the symmetry algebra of the CFT, and the dual lattices give their representations. If the lattice is self-dual, then the resulting CFT is actually chiral. The advantage of this is that the partition function is obtained from the lattice theta function only, without the need to also take into account the contribution of other representations of the symmetry algebra.

More recently Dymarsky and Shapere \cite{Dymarsky:2020qom} pointed out that quantum error correcting code of stabilizer type can be used to construct non-chiral CFTs. There has been a lot of follow-up work \cite{Buican:2021uyp, Yahagi:2022idq, Henriksson:2022dnu, Henriksson:2022dml,Kawabata:2022jxt, Kawabata:2023nlt, Alam:2023qac, Aharony:2023zit, Mizoguchi:2024ahp}. However, most of this work deals with fixed (and relatively small) central charges.
Not much has been done on the large $c$ limit of such theories.\cite{Angelinos:2022umf} discusses this limit, but only from an ensemble average point of view. 

In this article we return to the chiral case and use construction A to construct families of lattices. 
Such lattices cannot satisfy the lattice sparseness condition of \cite{Dymarsky:2020pzc}. In a first step, we therefore introduce an analogous sparseness condition for such code lattices. We then investigate the limit of various families of codes. In particular we investigate self-dual Reed-Muller codes and their lattices. We derive various bounds on their free energy and then conjecture that they satisfy the sparseness condition. We believe that this is a first step towards finding sparse lattices, in particular because such Reed-Muller codes can also be used to construct Barnes-Wall lattices, which in turn are a candidate for sparse lattices.

\emph{Note added for version 2:} After publication of our article, we became aware of reference \cite{Samorodnitsky2020AnIB}, in which improved bounds for certain codes are found. These new bounds are enough to establish $\theta_3$-sparseness for Reed-Muller codes. For completeness we added section~\ref{ss:doublytransitive}.

\section{Sparse Modular Forms}
\subsection{Sparseness conditions}
Let us begin with a quick review of sparseness conditions. Physically, such conditions are motivated by requiring that the holographic theory gives the right phase transitions expected for gravity. For instance, the original condition \cite{Hartman:2014oaa} ensured that the theory gives the expected Hawking-Page transition.

Mathematically, we impose this condition in the following way: We consider a family of modular objects $Z_n(\tau)$. For simplicity we specialize $\tau = iy$ to be purely imaginary. For $y\to\infty$, the behavior of $Z_n(y)$ is given by its leading term. We then call a family $Z_n(y)$ \emph{sparse} if in the limit $n\to \infty$, the behavior of its limit $Z(y)$ is given by the leading term for all $y>1$. That is, we exchange a limiting behavior in $y$ for a limiting behavior in $n$. In practice, to have a quantity for which this limit exists, we usually consider a rescaled logarithm of $Z_n(y)$ such as $\cF_n(y) = \frac1n \log Z_n(y)$. We will usually call this the free energy.

Let us make this more concrete by explaining how this works for the original sparseness condition for weakly holomorphic modular functions $Z_n(\tau)$. We take a Fourier expansion of the form $Z_n(\tau) = q^{-n} + \ldots$, so that for $y\to\infty$, that is for $q\to 0$, we have $\cF_n(y) = \frac1n \log Z_n(y) \to 2\pi y$. The sparseness condition on a family of such functions is thus
\be
\cF(y) = \lim_{n\to \infty} \frac1n \log Z_n(y) = 2\pi y \quad \textrm{for all }\ y>1\ .
\ee
Because of modularity, the behavior of $\cF(y)$ for $y<1$ is fixed by the behavior for $y>1$. In the case at hand we have $\cF(y) = \cF(1/y) = 2\pi/y$ for $y<1$.

For modular forms, this works similarly. In that case we want to consider modular forms $Z_n(\tau)$ of weight $n$ of the form $Z_n(\tau) = 1+ \ldots$. The sparseness condition is then
\be\label{latticesparse}
\cF(y) = \lim_{n\to \infty}\frac1n \log Z_n(y) = 0 \quad \textrm{for all }\ y>1\ ,
\ee
which gives the behavior $\cF(y) = \cF(1/y) = \log(iy)$ for $y<1$.  This is essentially the holomorphic version of the sparseness condition for lattice theta functions introduced in \cite{Dymarsky:2020pzc}.

\subsection{Extremal modular forms}\label{ss:EMF}
Let us now see if we can construct families of modular forms that satisfy (\ref{latticesparse}).
For this, let us start with the definition of an extremal modular form. We start with a quick review of modular forms. The modular group $SL(2,\Z)$ acts on the upper complex half plane via M\"obius transformations. A weight $k$ modular form $f$ is a function on $\mathbb{H}$ that satisfies: $f(\tau+1)=f(\tau)$ and $f\left(-\frac{1}{\tau}\right)=\tau^kf(\tau)$, and has a power series expansion in $q:=e^{2\pi i\tau}$.

A standard example of a modular form is the Eisenstein series $E_{2k}(\tau)$, which has weight $2k$ and has Fourier expansion
\be\label{E2k}
E_{2k}(\tau) = 1 - \frac{4k}{B_{2k}}\sum_{n=1}^\infty \sigma_{2k-1}(n)e^{2\pi i\tau n}\ ,
\ee
where the divisor sum function is defined as
\be
\sigma_l(r)=\sum_{d\mid r} d^l
\ee
and $B_{2k}$ is the Bernoulli numbers, with generating function:
\be
\sum_{n=0}^\infty \frac{B_{n}}{n!} x^{n} = \frac{x}{e^{x}-1}\ .
\ee
We note that the doubly even Bernoulli numbers $B_{4k}$ are negative, so that the doubly even Eisenstein series $E_{4k}$ have non-negative coefficients.
The space $M_k$ of modular forms of weight $k$, for even $k$, is a $\mathbb{C}-$vector space and has dimension 
\be
m_k=\dim M_k = \left\{ \begin{array}{cc} \lfloor \frac{k}{12} \rfloor & k \equiv 2 \mod 12 \\ 
\lfloor \frac{k}{12} \rfloor +1 & \textrm{otherwise} \end{array}
\right.
\ee
Maybe more usefully, a basis of $M$ is freely generated by $E_4$ and $\Delta$, with the weight given in the obvious way. Here $\Delta$ is given by
\be
\Delta = \frac{E_4^3-E_6^2}{1728} = q \prod_{i=1}^\infty (1-q^i)^{24}\ ,
\ee
and $\Delta$ has weight 12.
A basis for $M_k$ is then given by
\be
E_4^a, \ E_4^{a-3}\Delta, \ E_4^{a-6} \Delta^2, \ldots . \label{triangbasis}
\ee
This is an upper triangular basis in the sense that $E_4^a \Delta^b = q^b + \ldots$, where $4a+12b=k$.

This observation allows us to define extremal modular forms. The extremal modular form $F_k$ of weight $k$ is the unique modular form of that weight whose Fourier expansions starts as
\be
F_k(\tau) = 1 + 0q + 0q^2 + \ldots +0q^{m_k-1} + O(q^{m_k})\ .
\ee
That is, we fix the first $m_k$ Fourier coefficients to 1 and 0 respectively. Note that because of the triangular basis (\ref{triangbasis}), $F_k$ always exists and is unique. Moreover, because $E_4$ and $\Delta$ have integer coefficients, $F_k$ also has integer coefficients. 
As was first noted in \cite{MR0376536} and then worked out in more detail in \cite{MR2854563} however, the coefficients are not all positive starting at weight $k=81633$.

Not surprisingly, the family $F_k$ of extremal modular forms is sparse. Note however that because of the comment above, this does not follow directly from the criterion in \cite{Dymarsky:2020pzc}, which assumes that all coefficients of the form under consideration are positive. Instead we have to make a slightly more technical argument using the estimates for coefficients given in \cite{MR2854563}. These estimates are based on writing the extremal modular form as an Eisenstein series plus a correction function $h(\tau)$. Let us therefore first discuss sparseness of Eisenstein series.

\subsection{Eisenstein series}
Let us next consider the family of doubly even Eisenstein series $E_{4k}(\tau)$. These automatically have positive coefficients. However, beyond $k=1$ the coefficients are in general no longer integer, so that they cannot correspond to lattice theta functions. Nonetheless, let us show that they are sparse.

\begin{prop}\label{prop1}

The even index Eisenstein series $E_{2k}$ are a sparse family of modular forms.

\end{prop}

\begin{proof}

We need to show that $\lim_{k\rightarrow\infty}1/k\log E_{2k}(iy)=0$ if $y>1$. We will in fact show the stronger result
\be
\lim_{k\rightarrow\infty}\log(E_{2k}(iy))=0\ .
\ee 
From (\ref{E2k}) we see that we need to bound
\be
\log E_{2k}(iy)=\log\left(1-\frac{4k}{B_{2k}}\sum_{n=1}^\infty\sigma_{2k-1}(n)e^{-2\pi y n}\right)
\ee 
This will go to 0 if the second part of the log goes to 0, so that we need to investigate the sum
\be
\frac{4k}{|B_{2k}|}\sum_{n=1}^\infty\sigma_{2k-1}(n)e^{-2\pi y n}
\ee 
We have the following bounds and estimates for the quantities involved:
\be 
B_{2k}\sim4\sqrt{\pi k}\left(\frac{k}{\pi e}\right)^{2k}\textrm{\quad for }k\rightarrow\infty,\textrm{\qquad}\sigma_{2k-1}(n)\leq2n^{2k},
\ee 
Asymptotically, we can thus bound the sum as
\be
\frac{4k}{|B_{2k}|}\sum_{n=1}^\infty\sigma_{2k-1}(n)e^{-2\pi y n}\sim\frac{4k}{4\sqrt{\pi k}\left(\frac{k}{\pi e}\right)^{2k}}\sum_{n=1}^\infty 2n^{2k}e^{-2\pi y n}=2k^{1/2-2k}\pi^{2k-1/2}e^{2k}\sum_{n=1}^\infty n^{2k}e^{-2\pi y n}\ .
\ee 
The summands have a maximum at $n_0=\frac{k}{\pi y}$ of $(\pi e)^{-2k}(k/y)^{2k}$. We want to split the sum into a finite part that safely contains that maximum, and an infinite tail that is far enough from the maximum. Concretely, we split the sum at $n=2n_0^2$ and bound the two pieces individually:
\begin{enumerate}
\item For the finite part, we simply multiply the number of terms by the maximal term to obtain the bound
\begin{multline}
2k^{1/2-2k}\pi^{2k-1/2}e^{2k}\sum_{n=1}^{2n_0^2} n^{2k}e^{-2\pi y n}\\
\leq2k^{1/2-2k}\pi^{2k-1/2}e^{2k}\cdot 2n_0^2\cdot (\pi e)^{-2k}(k/y)^{2k}=4\pi^{-1/2}k^{5/2}y^{-2k-2}\ .
\end{multline}
This indeed goes to 0 for $k\to\infty$ as long as $y>1$. 

\item  The tail is given by 
\be
2k^{1/2-2k}\pi^{2k-1/2}e^{2k}\sum_{n=2n_0^2}^{\infty} n^{2k}e^{-2\pi y n}
\ee 
We claim that the summand can be bounded by $e^{-\pi y n}$. This is clearly the case for $n$ large enough, since polynomial growth is slower than exponential growth. To see that it already holds for $n\geq 2n_0^2$, write
$n^{2k}e^{-2\pi yn}=e^{2k\log(n)-2\pi y n}$; when $n=2n_0^2$, we have
\be
2k\log(2n_0^2)-\pi y (2n_0^2)=4k\log(n_0)+2k\log(2)-2\pi y n_0^2=2k\left(2\log\left(\frac{k}{\pi y}\right)+\log(2)-\frac{k}{\pi y}\right)
\ee 
which is negative for $k$ large enough; this establishes that the bound already holds for the first term.  Furthermore, we can consider the derivative
\be
\frac{d}{dn}2k\log(n)-2\pi y n=\frac{2k}{n}-\pi y\ .
\ee 
When $n\geq 2n_0^2$, this is bounded by
\be
\frac{2k}{2n_0^2}-2\pi y=\frac{\pi^2y^2}{k}-\pi y=-\pi y\left(1-\frac{\pi y}{k}\right)
\ee 
This is also negative for $k$ large enough, so for large enough $k$, we have indeed $2k\log(n)-\pi y n\leq0$ for all terms in the sum. 

We can thus bound the sum by the geometric series
\be
\sum_{n=2n_0^2}^\infty e^{-\pi y n}=\frac{e^{-\pi y 2n_0^2}}{1-e^{-\pi y}}=\frac{e^{-2k^2/(\pi y)}}{1-e^{-\pi y}}\ ,
\ee 
bounding the tail by
\be
2k^{1/2-2k}\pi^{2k-1/2}e^{2k}\frac{e^{-2k^2/(\pi y)}}{1-e^{-\pi y}}
\ee 
which goes to 0 as $k\to \infty$.
\end{enumerate}
\end{proof}

In particular, this shows that the doubly even Eisenstein series $E_{4k}$ are a sparse family with non-negative coefficients.
However, beyond weight 4 the coefficients will in general be non-integer, so that the $E_{4k}$ cannot correspond to a lattice theta series. Nonetheless,  this result may have implications on the sparseness of lattices: By the Siegel-Weil formula \cite{MR165033,MR223373}, the Eisenstein series $E_{4k}$ is the (appropriately weighted) average of the theta series of even unimodular lattices of dimension $8k$. The result thus seems to imply that in an appropriate sense, the `average' lattice is sparse.

We are now ready to give the proof that extremal modular forms are sparse:

\begin{prop}
The family of extremal modular forms $F_k$ are a sparse family of modular forms with integer (but eventually negative) coefficients. \end{prop}
\begin{proof}
Since $E_4$ and $\Delta$ have integer coefficients and each coefficient in the polynomial expansion of $F_k$ in $E_4$ and $\Delta$ is an integer, $F_k$ has integer coefficients. We thus only have to prove sparseness. To do this, we use the results of
\cite{MR2854563}. The idea is to write $F_k(\tau)=E_k(\tau)+h_k(\tau)$, where $h_k(\tau)=\sum_{n=1}^\infty b(n)q^n$ is some cusp form. 
As before, we then want to show that for $y>0$,
\be
\lim_{k\to\infty} F_k(iy) = \lim_{k\to\infty} (E_k(iy)+h_k(iy)) = 1\ . 
\ee
We already established this for $E_k(iy)$, so let us focus on $h_k(iy)$.
The results of \cite{MR2854563} establish that its coefficients can be bounded as $|b(n)|\leq2C_k n^{k/2}$, 
so that
\be
h_k(iy) = \sum_{n=1}^\infty 2 C_k n^{k/2} e^{-2\pi n y}\ 
\ee
with
\be 
C_k\leq\frac{(2\pi)^k}{(k-1)!}e^{28.466}\sqrt{\frac{k}{12}\log(k)}\left(1.0242382\frac{k}{12}\right)^{k/2}\ ,
\ee 
We deal with the sum in the same way as above, that is by identifying the maximal summand and then splitting the sum up into two pieces.
The summands have a maximum at $n_0=\frac{k}{4\pi y}$ of $2^{-k}(\pi e)^{-k/2}(k/y)^{k/2}$. Again splitting the sum at $2n_0^2$, up to constant factors the finite part is bounded by
\begin{multline}
\frac{(2\pi)^k k}{k!}\sqrt{k\log(k)}\left(1.0242382\frac{k}{12}\right)^{k/2}\left(\frac{k}{ y}\right)^2\left(2^{-k}(\pi e)^{-k/2}(k/y)^{k/2}\right)\\
\sim k^{-k+k/2+k/2} \left(\frac{1.024\ldots\cdot \pi e}{12 y}\right)^{k/2} 
\end{multline}
where in the second line we used Stirling's approximation for $k!$ and neglected all terms that do not grow at least exponentially in $k$. We thus find that this contribution indeed goes to 0 for $y>0.728893$. 

For the tail, we use the same as in proposition~\ref{prop1}  to establish that $e^{\frac{k}{2}\log(n)-2\pi y n}\leq e^{-\pi y n}$. A geometric series then bounds  bounding the tail by
\be
\frac{(2\pi)^k}{(k-1)!}e^{28.466}\sqrt{\frac{k}{12}\log(k)}\left(1.0242382\frac{k}{12}\right)^{k/2}\frac{e^{-2k^2/(\pi y)}}{1-e^{-\pi y}}
\ee 
which also goes to 0.

\end{proof}

\subsection{Floored Eisenstein series}
We saw that Eisenstein series are sparse, but cannot correspond to a lattice theta function because their coefficients are not integer. There is of course a fairly obvious way to turn the Eisenstein series $E_{k}$ into a modular form that has integer coefficients: We define  the \emph{floored Eisenstein series} $FE_{k}(\tau)$ as the unique modular form whose first $m_k$ coefficients are the floor of the first $m_k$ coefficients of $E_{k}$.\footnote{Note that this is different from taking the floor of \emph{all} coefficients of $E_k$: the resulting object would not be a modular form, whereas our definition by construction gives a modular form.}
By the same arguments as in section~\ref{ss:EMF}, $FE_{k}$ will automatically have integer coefficients. Not surprisingly, for doubly even $k$ it also has positive coefficients and is sparse, as the following proposition shows:

\begin{prop}
For large enough doubly even $k$,
the floored Eisenstein series $FE_{k}$ form a sparse family of modular forms with non-negative integer coefficients.
\end{prop}
\begin{proof}
Let us first show that all coefficients are non-negative. 
We write the floored Eisenstein series as the Eisenstein series minus a cusp form that floors the coefficients: $E_k(\tau)=\sum a(n)q^n$ and $h(\tau)=-\sum b(n)q^n$, whence $b(n)=a(n)\mod1$ for $0<n<m_k$. We want to show that $|b(n)|\leq|a(n)|$ for all $n$.

We can estimate $a(n)$ from below by using bounds on Bernoulli numbers, such as in \cite{MR2854563}, 
\be\label{anbound}
0.9997\frac{(2\pi)^k}{(k-1)!}n^{k-1}0.997\sim~\left(\frac{2\pi e}{k}\right)^kn^{k-1}p(k)\leq a(n)\ ,
\ee 
where the $\sim~$ comes from Stirling's inequality, and the $p(k)$ contains any terms that are not exponential with $k$ that we can safely disregard. 

On the other hand, we can bound $|b(n)|$ by using Theorem 1 in \cite{MR2854563}, which gives an expression for the coefficients $b(n)$ of $h(\tau)$ in terms of its first $m_k-1$ coefficients. We can write it as
\be\label{bnbound}
|b(n)|\leq C_k \sigma_0(n) n^{(k-1)/2} \leq 2 C_kn^{k/2}\ ,
\ee 
where we used the bound $\sigma_0(n) \leq 2n^{1/2}$. The $n$-independent coefficient
$C_k$ is given by 
\be
C_k=\sqrt{\log k}\left( 11 \sqrt{\sum_{m=1}^{m_k-1}\frac{|b(m)|^2}{m^{k-1}}} + \frac{e^{18.72}(41.41)^{k/2}}{k^{(k-1)/2}} \left|\sum_{m=1}^{m_k-1} b(m) e^{-7.288m}\right| \right)
\ee

The first thing we note is that for $1\leq n < m_k$, non-negativity is automatic by construction. The second thing we notice is that the bound (\ref{anbound}) grows much faster in $n$ than the bound (\ref{bnbound}). It is thus enough to establish the inequality for $ n=m_k \sim k/12$, as it will automatically hold for all larger $n$
as well.

We need to establish 
\be 
C_k\leq p(k)\left(\frac{2\pi e}{k}\right)^kn^{k/2-1}\ .
\ee 
Since the right hand side grows with $n$, it is enough to establish this for the smallest relevant value of $n$, namely $n=m_k\sim k/12$. In that case we need 
\be 
C_k\leq p(k)\left(\frac{2\pi e}{k}\right)^k\left(\frac{k}{12}\right)^{k/2-1}=p(k)k^{-k/2}\left(\frac{2\pi e}{\sqrt{12}}\right)^k
\ee 
Comparing with the first inequality, we will have the result if we can show that
\be \label{Acondition}
C_k\leq r(k)k^{-k/2}A^k
\ee 
with $A<\frac{2\pi e}{\sqrt{12}}$ and $r(k)$ containing any terms not exponential with $k$. In the following we will always neglect such terms.

The expression for $C_k$ in \cite{MR2854563}  contains two sums over $b(m)$. Since $b(m) = a(m) \mod 1$, we bound it by
\be
b(m) \leq \min(a(m), 1)\ ,
\ee
where in turn we estimate $a(m)$ as
\be
a(m) \sim (2\pi e/k)^k m^k\ .
\ee
In particular, $a(m)\sim 1$ for $m_0 \sim k/(2\pi e)$. We will argue that for both sums, the term with $m=m_0$ is the maximal term, and use that term to bound the sum.

For the first sum, note that $a(m)^2/m^k$ is increasing in $m$. On the other hand, we know that $|b(m)|<1$. It follows that the maximal term indeed arises for $m_0$ such that $b(m_0)\sim 1$,  giving a contribution
\be
\frac{1}{m_0^{k-1}}=k^{1-k}(2\pi e)^{k-1}\ .
\ee 
Taking into account the square root, the sum is bounded by
\be
k^{-k/2}(\sqrt{2\pi e})^k\ .
\ee 
Since $\sqrt{2\pi e}<2\pi e/\sqrt{12}$, (\ref{Acondition}) is satisfied.

Next consider the second sum. Using the same estimate as above, the summand $m^k e^{-7.288m}$ then has its maximum at $m =k/7.288 > k/12$. This means that in the summation region the summands again increase in $m$. The maximal term is thus again at $m_0\sim k/(2\pi e)$ with $|b(m_0)|\sim 1$. We can thus estimate the sum (again neglecting factors that are non-exponential in $k$) as
\be
\frac{(41.41)^{k/2}}{k^{k/2}}  e^{-7.288 m_0} 
\sim k^{-k/2} \left(e^{-7.288/2\pi e} \sqrt{41.41}\right)^k
\ee
The parentheses evaluates to $A=4.20\ldots$, which is slightly smaller than $\frac{2\pi e}{\sqrt{12}}$. The second sum thus also satisfies (\ref{Acondition}).
The conclusion is thus that for large enough doubly even $k$, the coefficients of $FE_{k}$ are indeed non-negative.

Now that we have established non-negativity, sparseness follows from the results of \cite{Dymarsky:2020pzc}: Since we already know that the Eisenstein series $E_{2k}$ are sparse and the floored Eisenstein series have by construction smaller coefficients for the first $m_k-1$ coefficients (the `light coefficients' in their language), the $FE_{k}$ are also sparse.

\end{proof}

\section{Lattices from codes}
\subsection{Construction A}
Let us now switch gears and construct actual lattices and their theta functions. We will use lattices that are constructed from codes. This will allow us to construct families of lattices and investigate their large $n$ limit. There are different types of such constructions. We will mostly focus on construction A. To this end, let us first introduce a few definitions on codes:

\defn 

\begin{enumerate}

\item A binary linear code $C$ is a linear subspace of $\Z_2^n$. $n$ is called the length of $C$. The rate of the code is $R=k/n$, where $k=\dim C$.

\item For $c\in C$, the weight of $c$, denoted $w(c)$, is the number of nonzero entries in $c$.
 $C$ is sometimes called an $[n,k,d]$-code, where $d$ is the minimum (non-zero) weight.

\item $C^\perp$ is the dual code to $C$ with respect to the inner product over $\Z_2$, that is the set of elements in $\Z_2^n$ that have even inner product with all elements of $C$.
If $C=C^\perp$ then $C$ is called self-dual. 

\item If $w(c)$ is even for all $c$, then $C$ is called an even code, and if $w(c)$ is doubly even - a multiple of four - for all $c$, then $C$ is called a doubly even code.

\item All self-dual codes are even. If $C$ is an even self-dual code, $C$ is called a Type I code, and if $C$ is a doubly even self-dual code, $C$ is called a Type II code. All type II codes have length divisible by 8.

\end{enumerate}

Finally let us introduce the weight enumerator of a code. The weight enumerator of $C$ is the polynomial 
\be
W_C(X,Y)=\sum_{k=0}^nA_kX^{n-k}Y^k\ ,
\ee
where $A_k$ is the number of codewords of weight $k$.

Now we introduce ways to construct lattices from codes.
\defn 
\textit{Construction A}: The Construction A lattice of $C$ is denoted $L_C$ and is given by $\frac{1}{\sqrt{2}}\rho^{-1}(C)$, where $\rho$ is reduction mod 2. This means that
\be\label{LconstA}
L_C = \left\{ \frac1{\sqrt{2}}(c+v) : c\in C, v \in (2\Z)^n \right\}\ ,
\ee
where by a slight abuse of notation we denote by $c$ also the codeword $c$ interpreted as a vector in $\R^n$.

\prop $L_C$ is a lattice in $\R^n$ with the following properties :

\begin{enumerate}
\item $\det L_C=2^{n-k-\frac{n}{2}}$.

\item $C\subseteq C^\perp$ iff $L_C$ is integral.

\item $C$ is doubly even iff $L_C$ is even.

\item $C$ is self-dual iff $L_C$ is self-dual.

\item A basis for $L_C$ is obtained  by taking a basis $c_1,...,c_k$ for $C$, extending it to a basis for $\Z_2^n$ in such a way that coercing the basis into $\mathbb{Z}^n$ results in a $\mathbb{Z}$-basis for $\mathbb{Z}^n$, and taking the $\rho^{-1}(c_i)$'s and $2\rho^{-1}(e_i)$'s.

\end{enumerate}
See \eg Chapter 5 of \cite{MR1662447} for this. The upshot is that type II codes give even self-dual lattice under construction A.

From (\ref{LconstA}) it is then straightforward to see that the theta series of $L_C$ is obtained by taking the Jacobi theta functions $\theta_3(0,q)=\sum_{k\in\Z}q^{k^2}$ and $\theta_2(0,q)=\sum_{k\in\Z}q^{(k+1/2)^2}$ and plugging them into the weight enumerator, 
\be
\Theta_{L_C} = W_C(\theta_3,\theta_2)\ .
\ee
One can obtain similar expressions for higher genus partition functions, which allows for instance to put constraints on possible code CFTs \cite{Henriksson:2021qkt}. We do not pursue this any further and focus on the genus 1 parition function in the following.

\subsection{$\theta_3$-sparseness}
Let us now discuss the limit of families of lattices that come from construction A.
The free energy of the Construction A lattice is given by
\be\label{FconstA}
\cF_n(y)=\frac{1}{n}\log\left(\sum_{k=0}^nA_k\theta_3^{n-k}\theta_2^k\right)
=\log(\theta_3)+\frac{1}{n}\log\left(\sum_{k=0}^nA_k\left(\frac{\theta_2}{\theta_3}\right)^k\right)\ .
\ee
We can immediately see that no such lattice can be sparse in the sense of (\ref{latticesparse}) because of the term $\log(\theta_3)$. This is simply a consequence of the fact that construction A adds lattice vectors $(2\Z)^n$, which lead to the contribution $\theta_3(y)$. We therefore introduce the following  generalization of the sparseness condition for construction A lattices:
Let $C_n$ be a family of doubly even self-dual codes of length $n$. We say the corresponding lattice $L_C$ is $\theta_3$-sparse if
\be
\lim_{n\to\infty} \cF_n(y) = \log(\theta_3(iy))\quad \textrm{for all }\ y>1\ ,
\ee

\subsection{Limits of various constructions}
Let us now compute the free energy of construction A lattices coming from various codes. As a warmup, let us consider some codes that are not type II.

\begin{example}The zero code has weight enumerator $X^n$, so that 
\be
\cF(y)= \lim_{n\rightarrow\infty}\frac1n\log(\theta_3(iy)^n)=\log(\theta_3(iy)).
\ee
\end{example}

\example The repetition code $R$, the code consisting of only the zero codeword and the codeword consisting of all 1's, has $W_R=X^n+Y^n$, giving
\be
\cF(y) = \max(\log(\theta_3(iy)),\log(\theta_2(iy))\ .
\ee

\example The universe code $\Z_2^n$ has $A_k=\binom{n}{k}$, so the weight enumerator is $(X+Y)^n$, and we attain
\be
\cF(y)=\log(\theta_3(iy)+\theta_2(iy))\ .
\ee

\example Let $B_n$ be the $(n/2)\times n$ matrix with entries from $\Z_2^n$ where the $k$th row has a pair of 1's in the $(2k-1)$th and $(2k)$th slots. Let $P_n$ be the linear code generated by $B_n$. $P_n$ is a Type I code, and we can find its weight enumerator by considering linear combinations of the rows of $B_n$; to get a weight $k$ codeword, we must have $k/2$ summands in the linear combination, so the number of codewords of weight $k$ is $\binom{n/2}{k/2}$. We thus have
\be 
W_{P_n}= \sum_{k=0}^{n/2}\binom{n/2}{k}X^{2k}Y^{n-2k}= (X^2+Y^2)^{n/2}\ ,
\ee 
giving
\be 
\cF(y)=\log(\theta_3(iy))+\frac{1}{2}\log\left(\left(\frac{\theta_2(iy)}{\theta_3(iy)}\right)^{2}+1\right)
\ee

\example The even universe code, consisting of all and only the codewords of $\Z_2^n$ with even weight, has weight enumerator $\frac{1}{2}((X+Y)^n+(X-Y)^n)$, so the limit of the free energy is 
\be
\cF(y)=\log(\max(|\theta_3(iy)+\theta_2(iy)|,|\theta_3(iy)-\theta_2(iy)|))
\ee

Let us give an example of a type II code.

\example Let $B_n$ be the $(n/2)\times n$ matrix with rows $e_k$, $k=0,\ldots n/2-1$, where the $e_k$ are given by
\bea
(e_0)_j &=& \left\{\begin{array}{cl} 1 &: j=0,n-2\\ 
0 &: j=1,n-1 \\
j \mod 2 &: \textrm{else}
\end{array}
\right.\\
(e_1)_j &=& j \mod 2 \\
(e_k)_j &=&\left\{\begin{array}{cl} 1 &: j=2k-2,2k-1,n-2,\textrm{ or } n-1\\ 
0 &: \textrm{else}
\end{array}
\right. \qquad k>1
\eea
where we take $j$ to run from $0$ to $n-1$.
For instance, $n=8$ gives
\be 
B_8=\begin{bmatrix}1&0&0&1&0&1&1&0\\0&1&0&1&0&1&0&1\\0&0&1&1&0&0&1&1\\0&0&0&0&1&1&1&1\end{bmatrix}
\ee 

\begin{lem}
Let $n$ be even, and let $S_n$ be the linear code generated by $B_n$. Then $S_n$ is a code of dimension $n/2$. It is type I if $n\mod 4 =0$, and type II if $n\mod 8 = 0$.
\end{lem}
\begin{proof}
We see that the $e_k$ form an upper triangular basis and are therefore linearly independent, so that $\dim S_n=n/2$. 

Assume $n\equiv0\mod4$. Define $E=\{e_2,...,e_{n/2-1}\}$. For all $x,y\in E$, clearly $x\cdot y=0$, since the only matching 1's are the last two. Moreover $e_1\cdot x=0$ for any $x\in E$, since $e_1$ only matches $x$ once in a middle pair and once in the last slot. We also have $e_i\cdot e_i=0$ for all $i$, as they are even codewords. 
It remains to show that $e_0\cdot e_1=e_0\cdot x=0$ for $x\in e$.
For $e_0\cdot e_1$, they match once in each middle pair, of which there is an even number.
For $e_0\cdot x$, they match once in the middle pairs and once in the second but last slot.
Hence we have $S_n\subseteq S_n^\perp$, which in particular establishes evenness. Moreover, since $\dim S_n=n/2$, the code is also self-dual.

Next assume $n\equiv0\mod8$. Note that if $u,v$ are doubly even codewords and $u\cdot v=0$, then $wt(u+v)=wt(u)+wt(v)-2j$, where $j$ is the number of matching 1's. Because of $u\cdot v=0$, $j$ has to be even, from which it follows that $u+v$ is also doubly even. Since for $n\equiv 0 \mod 8$, $S_n$ is generated by doubly even codewords, and moreover $S_n\subset S^\perp_n$, $S_n$ is doubly even and hence type II. 
\end{proof}

The following gives the weight enumerator of $S_n$: 
\begin{lem} The weight enumerator of $S_n$ is given by $A_k=\binom{n/2}{k/2}$ for $k\neq n/2$ and $A_k=\binom{n/2}{n/4}+2^{\frac{n}{2}-1}$ for $k=n/2$.

\end{lem}

\begin{proof} Define again $E=\{e_2,...,e_{n/2-1}\}$.

First, consider linear combinations $\langle E\rangle$ of the set $E$. When there is only one nontrivial summand, the result is weight 4. When there are only two nontrivial summands, two 0's in the middle will flip and two 1's at the end will flip, keeping the result at weight 4. Now, if one has only one nontrivial summand and then adds two more, it will add two pairs of 1's in the middle, and the end will flip twice, which is the same as not flipping, and so the weight increases by 4. The same occurs when one has two nontrivial summands and adds two more, so to get weight $k$, one must have $\frac{k}{2}-1$ or $\frac{k}{2}$ nontrivial summands, giving us $\binom{n/2-2}{k/2-1}+\binom{n/2-2}{k/2}$ combinations for weight $k$ codewords.

Next, consider vectors of the form $e_0 + \langle E\rangle$.  $e_0$ is weight $n/2$, and adding any vector in $\langle E\rangle$ will leave that weight unchanged, because each $e_i$ from $E$ will flip a 0 and 1 pair in the middle and a 0 and 1 pair at the end, keeping the weight same. The same argument works for $e_1 +\langle E\rangle$. In total we get $2^{\frac{n}{2}-1}$ such codewords of weight $n/2$.

Finally, consider $e_0+e_1+\langle E\rangle$. $e_0+e_1$ is weight 4, and $e_0+e_1+e_i$ for any $e_i\in e$ is also weight 4 for the same reasoning as in the first case. Similarly to before, adding pairs of $e_i$'s from $E$ to either of these cases will increase weight by 4, so to get weight $k$, add $\frac{k}{2}-2$ or $\frac{k}{2}-1$ $e_i$'s from $E$ to the respective scenarios, giving us $\binom{n/2-2}{k/2-2}+\binom{n/2-2}{k/2-1}$ more possibilities for weight $k$ codewords.

This exhausts all the possibilities, so collecting the results and applying Pascal's identity thrice yields the result.
\end{proof}

By reindexing with $k=4l$ and applying the result of the binomial theorem for $(1+x)$, we can use this result to write the theta series as
\be
\frac{\Theta(iy)}{\theta_3(iy)^n}=\frac{1}{2}\left(\left(-\left({\frac{\theta_2(iy)}{\theta_3(iy)}}\right)^2+1\right)^{n/2}+\left(\left({\frac{\theta_2(iy)}{\theta_3(iy)}}\right)^2+1\right)^{n/2}\right)+2^{\frac{n}{2}-1}\left({\frac{\theta_2(iy)}{\theta_3(iy)}}\right)^{n/2}
\ee 
For $y>1$, ${\frac{\theta_2(iy)}{\theta_3(iy)}}< 1/2.4$. To see this, note that $d (\theta_2(q)/\theta_3(q))/dq=\theta_2(q)\theta_4(q)^4/4q\theta_3(q)$, which is positive for real positive $q$. It follows that $\frac{\theta_2(iy)}{\theta_3(iy)}$ is monotonically decreasing in $y$. The bound then follows by evaluating it at $y=1$, which gives the (approximate) value above.
This means that the second term dominates over the first and last term. This leaves us with the limit of the free energy as
\be
\cF(y) = \log(\theta_3(iy))+\frac12\log
\left(1+\left({\frac{\theta_2(iy)}{\theta_3(iy)}}\right)^2\right)
\qquad y>1\ .
\ee 
Note that since $S(n)$ is a doubly even self-dual code, we can obtain the result for $y<1$ by modular invariance.

\section{Reed-Muller code lattices}

\subsection{Reed-Muller codes}

Let us now discuss a candidate for $\theta_3$-sparse lattices: the collection of Construction A lattices obtained from the self-dual Reed-Muller codes.

Let us first define these codes. We denote by
$RM(r,m)$ the $r$th order Reed-Muller code of length $2^m$. 
To construct it, consider the space of polynomials $p(x_1,\ldots,x_m)\in \mathbb{F}_2[x_1,...,x_m]$ in $m$ variables $x_i$ of degree $r$ or less over $\Z_2$. Each such polynomial $p$ then gives a codeword in $\Z_2^{2^m}$ by evaluating it for each element in $\Z_2^m$; that is, the $j$-th entry of the codeword is given by $p(b_1,\ldots,b_{m})$, where the $b_l$ are the digits of $j$ written in binary. 
We note that such polynomials  form a vector space over $\Z_2$ and that the evaluation map is a linear map. The codewords thus indeed give a linear code, namely the Reed-Muller code $R(r,m)$.

Note that since we are evaluating $p$ over $\Z_2$, we can work in $\mathbb{F}_2[x_1,...,x_m]/\langle x_1=x_1^2, ..., x_m=x_m^2\rangle.$, since $x_i$ and $x_i^2$ evaluate in the same way. That is, we can restrict ourselves to polynomials that contain at most one power of each $x_i$, so that the space of codewords is given by evaluating the polynomials
\be
RM(r,m) = \langle x_1^{n_1} \cdots x_m^{n_m} : 0\leq n_i \leq 1, \sum_{i=1}^m n_i \leq r \rangle\ .
\ee
Let us give a few properties of such codes; see \cite{MWS}
\begin{enumerate}
\item $RM(r,m)$ is even iff $r<m$. 
\item $RM(r,m)$ is doubly even iff $m>1$ and $r\leq\frac{m-1}{2}$. 
\item The dual code of $RM(r,m)$ is $RM(m-r-1,m)$; consequently, $RM(r,m)$ is self-dual iff $r=\frac{m-1}{2}$. Clearly, this is only possible if $m$ is odd.
\item In the notation above, $RM(r,m)$ is a $[2^m, \sum_{i=0}^r\binom{m}{i},2^{m-r}]$ code.\label{RMshortvec}
\end{enumerate}
As a consequence of the above, $RM\left(\frac{m-1}{2},m\right)$ for $m$ odd is a type II code. They thus form a natural candidate for our analysis.

Let us discuss point \ref{RMshortvec} above a bit further. It states that for Reed-Muller codes, the lightest (non-trivial) codewords have weight $2^{m-r}$. In particular, for our family of type II codes $RM\left(\frac{m-1}{2},m\right)$, the lightest codeword has weight
\be
w = 2^{\frac{m+1}2}\ .
\ee
In CFT language we would say that the code has a gap that grows parametrically in $m$. Any sparseness criterion is essentially the condition that light states, vectors or codewords have small multiplicity $a_k$. Reed-Muller codes having a large gap means that the very lightest vectors have multiplicity zero. It is thus natural ask if they are sparse, or, more precisely, $\theta_3$-sparse. In the following we will not fully prove this, but we will collect fairly compelling evidence that lattices constructed from $RM((m-1)/2,m)$ codes are indeed $\theta_3$-sparse. On one hand, this is an interesting result in its own right. On the other hand, since such RM codes can also be used to construct Barnes-Wall lattices \cite{MR1662447}, it may also have implications for the sparseness of Barnes-Wall lattices.

\subsection{Universal behavior of BEC code lattices}
Reed-Muller codes are actually an example of codes that achieve capacity for the binary erasure channel (BEC). For such codes, or for codes whose dual code $C^\perp$ achieves BEC capacity, bounds on the number of codewords $A_k$ of weight $k$ were obtained in \cite{8853269}. Let us first discuss bounds for BEC capacity code lattices, and then apply them to Reed-Muller codes.

We will write (\ref{FconstA}) as 
\be\label{cFy}
\cF_n(y)
=\log(\theta_3)+\frac{1}{n}\log\left(\sum_{k=0}^nA_k\alpha^k\right)\ ,
\ee
where we defined 
\be
\alpha(y):= \theta_2(iy)/\theta_3(iy)\ .
\ee
Note that $\alpha(y)$ is monotonically decreasing for $y>1$ and also $0< \alpha(y) < 1$ for $y>1$.  Establishing $\theta_3$-sparseness is thus a question of bounding the second term in (\ref{cFy}).

To do this, we will use the results of \cite{8853269}. For a code $C_n$ of length $n$ whose dual code $C^\perp_n$ achieves BEC capacity, its proposition 1.6 gives the following inequality for the $A_k$:
\be\label{BECbound}
A_k \leq 2^{o(n)} \left(\frac{1}{1-R} \right)^{(2\log 2)\min(k,n-k)}\ .
\ee
We can thus bound the sum in (\ref{cFy}) by
\be
2^{o(n)} \left(\frac{1-q^{n/2+1}}{1-q} + q^{n/2} \frac{1-\tilde q^{n/2+1}}{1-\tilde q}\right)\ ,
\ee
where 
\be
q := \frac{\alpha(y)}{(1-R)^{2\log 2}}\ , \qquad \tilde q = \alpha(y)(1-R)^{2\log 2}\ .
\ee
Here $R$ is the rate of the code. Because we are interested in $y>1$ and $R\leq 1$, we get the following bounds for the limit $\cF(y):= \lim_{n\to \infty} \cF_n(y)$:
\bea
\cF(y) = \log(\theta_3) &:& q<1\label{Fboundq} \\
\cF(y) \leq \log(\theta_3) + \frac12 \log q &:& q>1
\eea
It follows that $\cF(y) = \log(\theta_3)$ if $y>y_c$, where $y_c$ is fixed by the condition that $q(y_c)=1$.
Self-dual codes have $R=1/2$. In this case we find that
\be
y_c = 1.05\ldots\ .
\ee
This is not quite enough to prove that the code is $\theta_3$-sparse, but it comes very close.

Let refine this bound a little further.
Lemma 1.4 in \cite{8853269} introduces a function $F(\lambda,2)$ as  
\be 
F(\lambda,2)=\log_2\left(\frac{1}{|C_n|}\sum_{k=0}^nA_k(1-\theta)^k(1+\theta)^{n-k} \right)
\ee
for the parameter $0\leq\lambda\leq1$, where we defined $\theta=\lambda^{2\log2}$.
We can use this to express the sum in (\ref{cFy}) as
\be
\sum_{k=0}^nA_k\alpha^k=\sum_{k=0}^nA_k\left(\frac{1-\theta}{1+\theta}\right)^k = 2^{F(\lambda,2)}|C_n|(1+\theta)^{-n}
\ee
where we identified
\be
\theta = \frac{1-\alpha}{1+\alpha}\ .
\ee
In particular, $\theta$ increases monotonically with $y$.
In total this means we can express the free energy as
\be
\cF_n(y)=\log (\theta_3) + \frac1n F(\lambda,2) \log 2 +\frac1n\log |C_n| - \log(1+\theta) 
\ee
If $C_n$ achieves BEC capacity, then corollary 1.5 in \cite{8853269} asserts that there exists $\lambda=R+o(1)$ such that  $F(\lambda,2)\leq o(n)$. For self dual codes we have $R=1/2$ and $|C_n|=2^{n/2}$. This choice of $\lambda$ thus leads to a $\theta_c = 2^{-2\log 2}$ for which 
\be
\cF(y_c) = \log(\theta_3) + \log\frac{\sqrt 2}{1+2^{-2\log 2}} \ .
\ee
The free energy decreases monotonically in $y$, so this gives an upper bound that is valid for $y>y_c$.
With
\be
\alpha(y_c)= \frac{1-2^{-2\log 2}}{1+2^{-2\log 2}}\ .
\ee
we find $y_c \simeq 0.95 <1$, so the bound actually holds in the entire region of interest.

Collecting the various bounds, we find that for self-dual codes achieving BEC capacity we have
\bea
\cF(y) = \log(\theta_3) &:& y> 1.05 \label{firstregion} \\
\cF(y) \leq \log(\theta_3) + \frac12 \log \alpha(y)2^{2\log 2} &:&  1.05 > y > 1.02 \\
\cF(y) \leq \log(\theta_3) + 0.0226 &:&  1.02>y >1
\eea
This bound falls of course short of establishing that self-dual codes achieving BEC capacity are $\theta_3$-sparse. For this to hold, we would have to sharpen (\ref{firstregion}) to hold all the way up to $y=1$. It is an interesting question if such codes are indeed $\theta_3$-sparse and our estimates are simply not optimal, or if there are such codes which are not $\theta_3$-sparse. In the next section we will return to Reed-Muller codes and collect some additional evidence that they are $\theta_3$-sparse.

\subsection{Reed-Muller codes}
Reed-Muller codes achieve BEC capacity, so the above bounds apply to our family $RM(r,2r+1)$. Ideally we would like to establish that they are $\theta_3$-sparse. That is, we would like to sharpen the above bounds so that $\cF(y)=\log(\theta_3)$ all the way to $y>1$. To do this, we would want to sharpen the bound (\ref{BECbound}) on the multiplicities to something like
\be\label{RMbound}
A_w < 2^{o(n)} r^w
\ee
where $r\alpha(1) \leq 1$. Such a bound would establish (\ref{Fboundq}) for all $y>1$. Concretely, we need $r \leq \alpha(1)^{-1} =2.4$. 
Even though we won't establish (\ref{RMbound}) for the self-dual RM codes, we will collect some evidence in its favor by showing that it is valid for low values of $w$.

To this end, we will collect some information on the $a_w$ for the self-dual RM codes $RM(r,m=2r+1)$. For $RM(r,m)$, the lowest weight codewords have weight $d= 2^{m-r}$ with multiplicity \cite{8853269}
\be\label{aminweight}
A_{2^{m-r}}=2^r\prod_{i=0}^{r-1}\frac{2^{m-i}-1}{2^{r-i}-1}\ .
\ee
For our family, we take $r=(m-1)/2$ and $m$ large, so that the minimum weight is $d=2^{r+1}$.
We can then approximate (\ref{aminweight}) by
\be
A_d \sim  2^{r(r+2)} \sim 2^{(\log_2 d)^2}\ .
\ee
Comparing this to (\ref{RMbound}), we find that our RM codes are far from saturating it: There are no codewords of weights $1<w<d$, and $a_d$ grows sub-exponentially in $m$, unlike what is allowed by (\ref{RMbound}).

We can push this analysis a little bit further:
\cite{KT70} gives the multiplicities of all codewords up to weight $2d$, and \cite{KASAMI1976380} then does the same up to $2.5d$.
For $d<w<2d$, \cite{KT70} establishes that there are no codewords unless $w= 2d-2d/2^\mu$, where for self-dual Reed Muller codes $RM(r,2r+1)$ $\nu$ satisfies $1\leq \mu \leq \alpha = r $.
We find
\be
A_{w(\mu)} \sim 2^{-\mu ^2+2 \mu +r^2+2 \mu  r+2 r}
\ee
which is increasing in $\mu$. The heaviest codewords in this range have $\mu=r$, giving $w=2d-4$, with multiplicity 
\be\label{2dbound}
A_{2d-4}\sim 2^{2r(r+2)} \sim 2^{2(\log_2 2d)^2}\ .
\ee
For states with $d<w<2d$ we can use (\ref{2dbound}) as a uniform bound as
\be\label{wbound}
A_w < A_{2d-4} \sim 2^{2(\log_2 2d)^2} < 2^{2(\log_2 2w)^2}\ .
\ee
Even though we are using very rough estimates here, the growth in (\ref{wbound}) is still sub-exponential, and thus easily satisfies (\ref{RMbound}). We take this as evidence that the self-dual RM codes are indeed $\theta_3$-sparse.

\subsection{Doubly Transitive Self-Dual Codes}\label{ss:doublytransitive}

For doubly transitive linear codes of rate $R$, \cite{Samorodnitsky2020AnIB} gives an improved bound
\be
A_k\leq 2^{o(n)}\left(\frac{1}{2^{1-R}-1}\right)^k\ .
\ee
Self-dual codes have $R=1/2$, and we can thus repeat the argument below (\ref{BECbound}) by replacing
\be
q = \frac{\alpha(y)}{\sqrt{2}-1}\ 
\ee
and again obtain $y_c$ through $q(y_c)=1$. But in fact we have
\be
\alpha(1) =  \theta_2(i)/\theta_3(i)=\sqrt{2}-1\ .
\ee
(To see this,  write $\theta_2(\tau)^2/\theta_3(\tau)^2= \frac{\theta_3(\tau/2)^2-\theta_4(\tau/2)^2}{\theta_3(\tau/2)^2+\theta_4(
\tau/2)^2}$ and then use $\theta_3(i/2)^2=\sqrt{2}\theta_4(i/2)^2$.)
This establishes that $y_c=1$, meaning that the bound holds all the way to $y=1$. Doubly transitive self-dual codes are thus $\theta_3$-sparse.
$\theta_3$-sparseness of Reed-Muller codes follows from the fact that RM codes are triply transitive \cite{MWS}, and triply transitive in particular implies doubly transitive.

\section*{Acknowledgements}

We thank Gabriele Nebe for helpful discussions. We thank Anatoly Dymarsky for helpful discussions and comments on the draft. We thank anonymous referees for comments on the draft.
The work of CAK and AWR is supported in part by NSF Grant 2111748. CAK thanks the Pauli Center at ETH Zurich for hospitality, where part of this work was completed.

\bibliographystyle{ytphys}
\bibliography{refmain}

\def\cprime{$'$}
\providecommand{\href}[2]{#2}\begingroup\raggedright\begin{thebibliography}{10}

\bibitem{Witten:2007kt}
E.~Witten, ``{Three-Dimensional Gravity Revisited},''
\href{http://arxiv.org/abs/0706.3359}{{\ttfamily arXiv:0706.3359 [hep-th]}}.

\bibitem{Hawking:1982dh}
S.~W. Hawking and D.~N. Page, ``{Thermodynamics of Black Holes in anti-De Sitter Space},''
\href{http://dx.doi.org/10.1007/BF01208266}{{\em Commun. Math. Phys.} {\bfseries 87} (1983) 577}.

\bibitem{Hartman:2014oaa}
T.~Hartman, C.~A. Keller, and B.~Stoica, ``{Universal Spectrum of 2d Conformal Field Theory in the Large c Limit},'' \href{http://dx.doi.org/10.1007/JHEP09(2014)118}{{\em JHEP} {\bfseries 09} (2014) 118},
\href{http://arxiv.org/abs/1405.5137}{{\ttfamily arXiv:1405.5137 [hep-th]}}.

\bibitem{Dey:2024nje}
I.~Dey, S.~Pal, and J.~Qiao, ``{A universal inequality on the unitary 2D CFT partition function},'' \href{http://dx.doi.org/10.1007/JHEP07(2025)163}{{\em JHEP} {\bfseries 07} (2025) 163}, \href{http://arxiv.org/abs/2410.18174}{{\ttfamily arXiv:2410.18174 [hep-th]}}.

\bibitem{Dymarsky:2020pzc}
A.~Dymarsky and A.~Shapere, ``{Comments on the holographic description of Narain theories},'' \href{http://dx.doi.org/10.1007/JHEP10(2021)197}{{\em JHEP} {\bfseries 10} (2021) 197}, \href{http://arxiv.org/abs/2012.15830}{{\ttfamily arXiv:2012.15830 [hep-th]}}.

\bibitem{MR209983}
J.~Leech, ``Notes on sphere packings,'' \href{https://doi.org/10.4153/CJM-1967-017-0}{{\em Canadian J. Math.} {\bfseries 19} (1967) 251--267}.

\bibitem{MR1662447}
J.~H. Conway and N.~J.~A. Sloane, \href{http://dx.doi.org/10.1007/978-1-4757-6568-7}{{\em Sphere packings, lattices and groups}}, vol.~290 of {\em Grundlehren der Mathematischen Wissenschaften [Fundamental Principles of Mathematical Sciences]}.
\newblock Springer-Verlag, New York, third~ed., 1999.
\newblock \url{http://dx.doi.org/10.1007/978-1-4757-6568-7}.
\newblock With additional contributions by E. Bannai, R. E. Borcherds, J. Leech, S. P. Norton, A. M. Odlyzko, R. A. Parker, L. Queen and B. B. Venkov.

\bibitem{MR1489922}
G.~Nebe, ``Some cyclo-quaternionic lattices,'' \href{https://doi.org/10.1006/jabr.1997.7163}{{\em J. Algebra} {\bfseries 199} no.~2, (1998) 472--498}.

\bibitem{MR2999133}
G.~Nebe, ``An even unimodular 72-dimensional lattice of minimum 8,'' {\em J. Reine Angew. Math.} {\bfseries 673} (2012) 237--247.

\bibitem{MR3225314}
G.~Nebe, ``A fourth extremal even unimodular lattice of dimension 48,'' \href{http://dx.doi.org/10.1016/j.disc.2014.05.011}{{\em Discrete Math.} {\bfseries 331} (2014) 133--136}.

\bibitem{MR0376536}
C.~L. Mallows, A.~M. Odlyzko, and N.~J.~A. Sloane, ``Upper bounds for modular forms, lattices, and codes,'' \href{https://doi.org/10.1016/0021-8693(75)90155-6}{{\em J. Algebra} {\bfseries 36} no.~1, (1975) 68--76}.

\bibitem{MR2854563}
P.~Jenkins and J.~Rouse, ``Bounds for coefficients of cusp forms and extremal lattices,'' \href{https://doi.org/10.1112/blms/bdr030}{{\em Bull. Lond. Math. Soc.} {\bfseries 43} no.~5, (2011) 927--938}.

\bibitem{MR165033}
A.~Weil, ``Sur certains groupes d'op\'{e}rateurs unitaires,'' \href{https://doi.org/10.1007/BF02391012}{{\em Acta Math.} {\bfseries 111} (1964) 143--211}.

\bibitem{MR223373}
A.~Weil, ``Sur la formule de {S}iegel dans la th\'{e}orie des groupes classiques,'' \href{https://doi.org/10.1007/BF02391774}{{\em Acta Math.} {\bfseries 113} (1965) 1--87}.

\bibitem{MR2209183}
G.~Nebe, E.~M. Rains, and N.~J.~A. Sloane, {\em Self-dual codes and invariant theory}, vol.~17 of {\em Algorithms and Computation in Mathematics}.
\newblock Springer-Verlag, Berlin, 2006.

\bibitem{MR285994}
J.~Leech and N.~J.~A. Sloane, ``Sphere packings and error-correcting codes,'' \href{https://doi.org/10.4153/CJM-1971-081-3}{{\em Canadian J. Math.} {\bfseries 23} (1971) 718--745}.

\bibitem{Dymarsky:2020qom}
A.~Dymarsky and A.~Shapere, ``{Quantum stabilizer codes, lattices, and CFTs},'' \href{http://dx.doi.org/10.1007/JHEP03(2021)160}{{\em JHEP} {\bfseries 21} (2020) 160}, \href{http://arxiv.org/abs/2009.01244}{{\ttfamily arXiv:2009.01244 [hep-th]}}.

\bibitem{Buican:2021uyp}
M.~Buican, A.~Dymarsky, and R.~Radhakrishnan, ``{Quantum codes, CFTs, and defects},'' \href{http://dx.doi.org/10.1007/JHEP03(2023)017}{{\em JHEP} {\bfseries 03} (2023) 017}, \href{http://arxiv.org/abs/2112.12162}{{\ttfamily arXiv:2112.12162 [hep-th]}}.

\bibitem{Yahagi:2022idq}
S.~Yahagi, ``{Narain CFTs and error-correcting codes on finite fields},'' \href{http://dx.doi.org/10.1007/JHEP08(2022)058}{{\em JHEP} {\bfseries 08} (2022) 058}, \href{http://arxiv.org/abs/2203.10848}{{\ttfamily arXiv:2203.10848 [hep-th]}}.

\bibitem{Henriksson:2022dnu}
J.~Henriksson, A.~Kakkar, and B.~McPeak, ``{Narain CFTs and quantum codes at higher genus},'' \href{http://dx.doi.org/10.1007/JHEP04(2023)011}{{\em JHEP} {\bfseries 04} (2023) 011}, \href{http://arxiv.org/abs/2205.00025}{{\ttfamily arXiv:2205.00025 [hep-th]}}.

\bibitem{Henriksson:2022dml}
J.~Henriksson and B.~McPeak, ``{Averaging over codes and an SU(2) modular bootstrap},'' \href{http://dx.doi.org/10.1007/JHEP11(2023)035}{{\em JHEP} {\bfseries 11} (2023) 035}, \href{http://arxiv.org/abs/2208.14457}{{\ttfamily arXiv:2208.14457 [hep-th]}}.

\bibitem{Kawabata:2022jxt}
K.~Kawabata, T.~Nishioka, and T.~Okuda, ``{Narain CFTs from qudit stabilizer codes},'' \href{http://dx.doi.org/10.21468/SciPostPhysCore.6.2.035}{{\em SciPost Phys. Core} {\bfseries 6} (2023) 035}, \href{http://arxiv.org/abs/2212.07089}{{\ttfamily arXiv:2212.07089 [hep-th]}}.

\bibitem{Kawabata:2023nlt}
K.~Kawabata and S.~Yahagi, ``{Fermionic CFTs from classical codes over finite fields},'' \href{http://dx.doi.org/10.1007/JHEP05(2023)096}{{\em JHEP} {\bfseries 05} (2023) 096}, \href{http://arxiv.org/abs/2303.11613}{{\ttfamily arXiv:2303.11613 [hep-th]}}.

\bibitem{Alam:2023qac}
Y.~F. Alam, K.~Kawabata, T.~Nishioka, T.~Okuda, and S.~Yahagi, ``{Narain CFTs from nonbinary stabilizer codes},'' \href{http://dx.doi.org/10.1007/JHEP12(2023)127}{{\em JHEP} {\bfseries 12} (2023) 127}, \href{http://arxiv.org/abs/2307.10581}{{\ttfamily arXiv:2307.10581 [hep-th]}}.

\bibitem{Aharony:2023zit}
O.~Aharony, A.~Dymarsky, and A.~D. Shapere, ``{Holographic description of Narain CFTs and their code-based ensembles},'' \href{http://dx.doi.org/10.1007/JHEP05(2024)343}{{\em JHEP} {\bfseries 05} (2024) 343}, \href{http://arxiv.org/abs/2310.06012}{{\ttfamily arXiv:2310.06012 [hep-th]}}.

\bibitem{Mizoguchi:2024ahp}
S.~Mizoguchi and T.~Oikawa, ``{Unifying error-correcting code/Narain CFT correspondences via lattices over integers of cyclotomic fields},'' \href{http://dx.doi.org/10.1016/j.physletb.2025.139308}{{\em Phys. Lett. B} {\bfseries 862} (2025) 139308}, \href{http://arxiv.org/abs/2410.12488}{{\ttfamily arXiv:2410.12488 [hep-th]}}.

\bibitem{Angelinos:2022umf}
N.~Angelinos, D.~Chakraborty, and A.~Dymarsky, ``{Optimal Narain CFTs from codes},'' \href{http://dx.doi.org/10.1007/JHEP11(2022)118}{{\em JHEP} {\bfseries 11} (2022) 118}, \href{http://arxiv.org/abs/2206.14825}{{\ttfamily arXiv:2206.14825 [hep-th]}}.

\bibitem{Samorodnitsky2020AnIB}
A.~Samorodnitsky, ``An improved bound on $\ell_q$ norms of noisy functions,'' \href{http://arxiv.org/abs/2010.02721}{{\ttfamily arXiv:2010.02721 [cs.IT]}}.

\bibitem{Henriksson:2021qkt}
J.~Henriksson, A.~Kakkar, and B.~McPeak, ``{Classical codes and chiral CFTs at higher genus},'' \href{http://dx.doi.org/10.1007/JHEP05(2022)159}{{\em JHEP} {\bfseries 05} (2022) 159}, \href{http://arxiv.org/abs/2112.05168}{{\ttfamily arXiv:2112.05168 [hep-th]}}.

\bibitem{MWS}
F.~J. MacWilliams and N.~J.~A. Sloane, {\em The Theory of Error Correcting Codes}, vol.~16 of {\em North Holland Mathematical Library}.
\newblock North Holland Publishing Company, third~ed., 1977.
\newblock Chapters 13 and 15.

\bibitem{8853269}
A.~Samorodnitsky, ``An upper bound on $\ell_q$ norms of noisy functions,'' \href{http://arxiv.org/abs/1809.09696}{{\em IEEE Transactions on Information Theory} {\bfseries 66} no.~2, (2020) 742--748}, \href{http://arxiv.org/abs/1809.09696}{{\ttfamily 1809.09696}}.

\bibitem{KT70}
T.~Kasami and N.~Tokura, ``On the weight structure of reed-muller codes,'' {\em IEEE Transactions on Information Theory} (1970) .

\bibitem{KASAMI1976380}
T.~Kasami, N.~Tokura, and S.~Azumi, ``On the weight enumeration of weights less than 2.5d of reed—muller codes,'' \href{https://www.sciencedirect.com/science/article/pii/S0019995876903557}{{\em Information and Control} {\bfseries 30} no.~4, (1976) 380--395}.

\end{thebibliography}\endgroup

\end{document}